\documentclass[times]{mcp-acm}

\usepackage{amsmath}    
\usepackage{graphicx}   

\usepackage{tikz}       
\usepackage{amsthm}     
\usepackage{amsfonts}   
\usepackage{caption}    
\usepackage{subcaption} 
\usepackage{enumitem}   

\makeatletter
\newcommand{\chapterauthor}[1]{%
	{\parindent0pt\vspace*{-5pt}%
		\linespread{1.1}\large\scshape#1%
		\par\nobreak\vspace*{35pt}}
	\@afterheading%
}
\makeatother

\newtheorem{lemma}{Lemma}
\newtheorem{corollary}{Corollary}
\newtheorem{proposition}{Proposition}
\newtheorem{definition}{Definition}

\renewcommand{\thesection}{\arabic{section}}

\begin{document}
\chapter*{Instrumental Variables with Treatment-Induced Selection: Exact Bias Results}
\begin{LARGE}Felix Elwert\end{LARGE}\\
\begin{large}University of Wisconsin--Madison\end{large}\vspace{0.3cm}\\
\begin{LARGE}Elan Segarra\end{LARGE}\\
\begin{large}University of Wisconsin--Madison\end{large}\\
\begin{quote} \begin{small}
Instrumental variables (IV) estimation suffers selection bias when the analysis conditions on the treatment. Judea Pearl's [2000:248] early graphical definition of instrumental variables explicitly prohibited conditioning on the treatment. Nonetheless, the practice remains common. 
In this paper, we derive exact analytic expressions for IV selection bias across a range of data-generating models, and for various selection-inducing procedures. 
We present four sets of results for linear models.
First, IV selection bias depends on the conditioning procedure (covariate adjustment vs. sample truncation). Second, IV selection bias due to covariate adjustment is the limiting case of IV selection bias due to sample truncation.
Third, in certain models, the IV and OLS estimators under selection bound the true causal effect in large samples.
Fourth, we characterize situations where IV remains preferred to OLS despite selection on the treatment.
These results broaden the notion of IV selection bias beyond sample truncation, replace prior simulation findings with exact analytic formulas, and enable formal sensitivity analyses.
\end{small} \end{quote}

\section{Introduction}

Instrumental variables (IV) analysis is a popular approach for identifying causal effects when the treatment is confounded by omitted variables. IV analysis rests on two main assumptions: that the instrument is associated with the treatment (``relevance"), and that the instrument is associated with the outcome only via the effect of treatment on the outcome (``exclusion"). The exclusion assumption is the sticking point of many empirical applications, because it requires theoretical justification and is testable only to a very limited degree (e.g., \citealt{BalkePearl1997}; \citealt{RichardsonRobins2010}). 

One type of exclusion violation that has recently gained attention is selection bias (e.g., \citealt{Swanson2015}; \citealt{Engberg2014}; \citealt{ErtefaieEtAl2016}; \citealt{HughesEtAl2019}; \citealt{CananEtAl2017}; \citealt{GkatzionisBurgess2018};
\citealt{MogstadtWiswall2012}). 
We say that IV analysis suffers selection bias when conditioning (rather than not conditioning) on some variable violates the exclusion assumption. One particularly important case is \emph{treatment-induced IV selection bias}: whenever treatment is confounded by unobservables, conditioning on a variable that has been affected by treatment induces bias. Judea \citet[p.~248]{Pearl2000} recognized this problem and presented the first definition of instrumental variables that outright prohibits conditioning on  variables affected by treatment. Despite Pearl's warning, however, conditioning on such ``descendants" of treatment remains common in IV analysis. 

Past research on treatment-induced IV selection bias (\citealt{Swanson2015}; \citealt{HughesEtAl2019}; \citealt{CananEtAl2017}; \citealt{GkatzionisBurgess2018}) is limited in two respects. First, it has focused on IV selection bias induced by sample truncation, which occurs when observations are excluded from the sample.\footnote{ Some studies have proposed corrections, bounds, or sensitivity analyses for IV selection bias in certain truncation scenarios (e.g., \citealt{MogstadtWiswall2012}; \citealt{Engberg2014}; \citealt{CananEtAl2017}; \citealt{VansteelandtEtAl2018}; \citealt{GkatzionisBurgess2018}; \citealt{HughesEtAl2019}). These approaches often rely on knowing the selection probability of both the observed and the truncated observations. }  This focus neglects that other conditioning procedures, such as covariate adjustment, can also induce selection bias. Second, in situations where consistent estimators are not readily available, the literature characterizes the size and sign of IV selection bias by simulation. Without analytic bias expressions, however, it is unclear which stylized facts resulting from simulation studies hold generically. 

This paper makes two main contributions. First, we derive analytic expressions for treatment-induced IV selection bias for a range of different data-generating models. Second, we compare the biases resulting from two different selection-inducing conditioning procedures: sample truncation and covariate adjustment.  
For tractability, we focus on linear models with homogeneous (constant) effects and normal errors. 

We highlight several results. First, the selection procedure matters. Within a given data-generating model, selection by truncation and selection by covariate adjustment introduce quantitatively different biases into IV analysis. Second, selection bias by adjustment is the limiting case of selection bias by truncation. Third, in certain models, the IV and OLS estimators with selection bound the true causal effect in large samples. Fourth, our analytic bias expressions characterize the models in which IV is less biased than OLS, which obtains when treatment does not exert an extreme effect on selection.

The rest of the paper proceeds as follows. Section \ref{sec:causal graphs} reviews basic facts about directed acyclic graphs for linear models. Section \ref{sec:IV} defines instrumental variables in econometric and graphical notation. Section \ref{sec:qualitative} describes conditions under which selection violates the IV exclusion assumption and defines IV estimation under selection by truncation and covariate adjustment. Section \ref{sec:quantitative} presents analytic expressions for the bias in IV and OLS estimators over a range of models with treatment-induced selection by trunction and by covariate adjustment. 
Section \ref{sec:conclusion} concludes.

\section{Causal Graphs}\label{sec:causal graphs}
The challenge of selection bias in IV analysis is transparently communicated with graphical causal models (\citealt{Pearl2009}; \citealt{MaathuisEtAl2018}).
Here, we review the basics.
A \emph{causal graph} represents the structural equations of the data-generating model.
Causal graphs consist of \emph{nodes} representing variables and \emph{directed edges} representing direct causal effects.
Causal graphs are assumed explicitly to display the observed and unobserved common causes of all variables.
By convention, causal graphs do not explicitly display the idiosyncratic shocks that affect individual variables.

\begin{figure}
	\centering
	\begin{subfigure}[b]{0.48\linewidth}
	    \centering
\def\xscale{1.2}
\def\yscale{1.2}
\tikzstyle{DAGarrow} = [-latex]
\begin{tikzpicture}[scale=1.2, transform shape]
	\node (z) at (-2*\xscale,0) {$Z$};
	\node (t) at (0,0) {$T$};
	\node (y) at (2*\xscale,0) {$Y$};
	\node (u) at (1*\xscale,1*\yscale) {$U$};
	\node (s) at (0,-1*\yscale) {$S$};

	\draw [DAGarrow] (z) -- node [above] {$\pi$} (t);
	\draw [DAGarrow] (t) -- node [above] {$\beta$} (y);
	\draw [DAGarrow] (u) -- node [above left] {$\delta_1$} (t);
	\draw [DAGarrow] (u) -- node [above right] {$\delta_2$} (y);
	\draw [DAGarrow] (t) -- node [left] {$\gamma$} (s);
\end{tikzpicture}
	    \caption{}
	    \label{fig:IVDAGBaseline}
	\end{subfigure}
	\begin{subfigure}[b]{0.48\linewidth}
	    \centering
	    \begin{align*}
            Z & = \varepsilon_Z \\
            U & = \varepsilon_U \\
            T & = \pi Z+\delta_{1}U+\varepsilon_{T} \\
            S & = \gamma T+\varepsilon_{S} \\
            Y & = \beta T+\delta_{2}U+\varepsilon_{Y}
        \end{align*}
	    \caption{}
	    \label{fig:IVLSEM}
	\end{subfigure}
	\caption{IV scenario where the selection variable is a function of treatment alone, equivalently displayed as a causal graph (\subref{fig:IVDAGBaseline}) and as a linear structural equations model (\subref{fig:IVLSEM}).} 
	\label{fig:IVBaselineModel}
\end{figure}
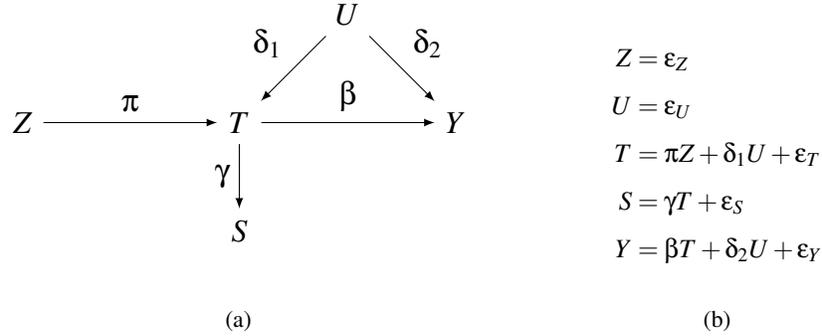

Throughout, we assume that the causal graphs represent linear  data-generating models with homogeneous effects and normally distributed errors.\footnote{Some results do not rest on the joint normality assumption, but our results on IV selection bias with truncation do.}
Without loss of generality, we further assume that all variables are standardized to have mean zero and unit variance.
The direct causal effect of one variable on another variable in such models is given by its \emph{path parameter}, which is bounded by [-1,1].
For example, the causal graph in Figure \ref{fig:IVDAGBaseline} represents the linear structural equations model given in Figure \ref{fig:IVLSEM}, with path parameters $ \pi, \beta, \gamma, \delta_1, $ and $ \delta_2 $.
For each variable $ V \in \{Z, U, T, S, Y\} $ the idiosyncratic shocks are marginally independent and normally distributed, $  \varepsilon_V\sim N(0,\sigma_V^2) $, with variance $ \sigma_V^2 $ scaled so that each $ V \sim N(0,1) $. 
Since $ U $ is unobserved, the \emph{structural error} term on $ Y $ in econometric terminology is $ \omega_Y = \delta_2U + \varepsilon_{Y} $.
Notice that $ T $ is correlated with the structural error, $ Cov(T,\omega_Y)\neq 0 $, because both depend on the unobserved confounder, $ U $. 

Under mild conditions to avoid knife-edge cases, simple rules determine the covariance structure of data generated by a model (\citealt{Pearl2009}). 
The notions of paths, collider variables, and descendants play a central role in these rules.
A \emph{path} is an acyclic sequence of adjacent arrows between two variables, regardless of the direction of the arrows. 
In a \emph{causal path} from treatment to outcome, all arrows point toward the outcome. 
In a \emph{non-causal}, or \emph{spurious}, path between treatment and outcome, at least one arrow points away from the outcome. 
A variable is called a \emph{collider} with respect to a specific path if it receives two inbound arrows on the path. For example, $ T $ is a collider on the path $ Z\rightarrow T\leftarrow U\rightarrow Y $. 
The \emph{descendant set} of a variable contains all variables directly and indirectly caused by it, e.g. $ desc(T) = \{S, Y\} $ in Figure \ref{fig:IVDAGBaseline}. 

Two variables are statistically independent if all paths between them are closed; and two variables are statistically associated if there is at least one open path between them (\citealt{VermaPearl1988}). 
A path is \emph{closed} (does not “transmit” association) if either (a) it contains a collider and neither the collider nor any of its descendants are conditioned on, or (b) it contains a non-collider that is conditioned by exact stratification. 
A path is \emph{open} (does “transmit” association) iff it is not closed (\citealt{Pearl1988}). 
Importantly, when a path contains only one collider, then conditioning on this collider, or any of its descendants, opens this path. 

The marginal covariance between two variables in a linear model with standardized variables is given by Wright's [1934] rule as the sum of the product of the path parameters on the open paths that connect the variables. 
For example, the marginal covariance between $ Z $ and $ Y $ in Figure \ref{fig:IVDAGBaseline} is $ Cov(Y,Z) = \pi\beta $, because the path $ Z\rightarrow T\rightarrow Y $ is the only open path (the other path, $ Z\rightarrow T\leftarrow U \rightarrow Y $, is closed by the unconditioned collider $ T $). 
The conditional covariance between variables $ A $ and $ B $, after adjusting for some covariate $ C $, is $ Cov(A,B|C)=Cov(A,B)-Cov(A,C) Cov(B,C) $. 
The novel bias results in this paper hinge on deriving conditional covariances when \emph{truncating} the sample as a function of $ C $.

\section{Instrumental Variables}\label{sec:IV}
Let $ T $ be the treatment variable of interest, $ Y $ be the outcome, $ Z $ be the candidate instrumental variable, and $ \mathbf{X} $ be a set of covariates. Econometrically, an instrumental variable is defined by two assumptions. 
\begin{definition}
    A variable, $ Z $, is called an instrumental variable for the causal effect of $ T $ on $ Y $, $ \beta $, if, conditional on the set of covariates $ \mathbf{X} $ (which may be empty),
    \begin{description}[labelindent=0.5cm]
    	\item[E1:] $ Z $ is associated with $ T $, $ Cov(Z,T|\mathbf{X})\neq 0 $, 
    	\item[E2:] $ Z $ is not associated with the structural error term, $ \omega_Y $, on $ Y $, $ Cov(Z,\omega_Y |\mathbf{X})=0 $.
    \end{description}
\end{definition}
Assumption E1 is called \emph{relevance}, and assumption E2 is called \emph{exclusion}. 
\citet{Pearl2001} provides a graphical definition.
\begin{definition}
    A variable, $ Z $, is called an instrumental variable for the causal effect of $ T $ on $ Y $, $ \beta $, if, conditional on the set of covariates $ \mathbf{X} $ (which may be empty),
    \begin{description}[labelindent=0.5cm]
   	\item[G1:] There is at least one open path from $ Z $ to $ T $ conditional on $ \mathbf{X} $,
   	\item[G2:] $ \mathbf{X} $ does not contain descendants of $ Y $, $ \mathbf{X}\cap desc(Y) = \emptyset$,
    	\item[G3:] There is no open path from $ Z $ to $ Y $ conditional on $ \mathbf{X} $, other than those paths that terminate in a causal path from $ T $ to $ Y $. 
    \end{description}
\end{definition}

Assumption G1 defines \emph{relevance}, and assumptions G2 and G3 together define \emph{exclusion}. 
We say that a candidate instrumental variable is ``valid" if it is relevant and excluded, and ``invalid" otherwise. 
For example, in Figure \ref{fig:IVDAGBaseline}, $ Z $ is a valid instrument without conditioning on $ S $, since $ Z $ is relevant (associated with $ T $) by the open path $ Z\rightarrow T $, and $ Z $ is excluded (unassociated with the structural error term on $ Y $) since the only open path from $ Z $ to $ Y $, $ Z\rightarrow T\rightarrow Y $, terminates in the causal effect of $ T $ on $ Y $. 
When $ Z $ is a valid instrumental variable, then the standard IV estimator, given by the sample analog of
\[ \beta_{IV} = \frac{Cov(Y,Z|\mathbf{X})}{Cov(T,Z|\mathbf{X})} , \]
is consistent for the causal effect of $ T $ on $ Y $ in linear and homogeneous models. 
The numerator of this estimator is called the \emph{reduced form} and the denominator is called the \emph{first stage.} 
The behavior of this IV estimator is the focus of this paper. 
For simplicity, we will henceforth write $ \beta_{IV} $ and $ \beta_{OLS} $ to refer to the probability limits (as the sample size tends to infinity) of the standard IV and OLS estimators, respectively.

\section{Selection Bias in IV: Qualitative Analysis}\label{sec:qualitative}
We say that the IV estimator suffers selection bias when conditioning on some variable violates the exclusion assumption.
For example, conditioning on a variable that opens a path between $Z$ and $Y$ that does not terminate in the causal effect of $T$ on $Y$ violates exclusion both in the sense of G3 and E2. 
\citet{HughesEtAl2019} catalogue several models in which selection violates exclusion. 

We focus on the IV selection bias that results from conditioning on a descendant of $ T $, $ S \in desc(T) $. 
For example, in Figure \ref{fig:IVDAGBaseline}, conditioning on $ S $ invalidates the use of $ Z $ as an instrumental variable, because $ T $ is the only collider variable on the path $ Z\rightarrow T\leftarrow U\rightarrow Y $, and conditioning on $S$ as the descendant of the collider $T$ opens this path. 
The association ``transmitted" by this open path overtly violates the exclusion condition G3 and similarly violates the exclusion condition E2, since $ \omega_Y $ is a function of $ U $. This rationalizes why Pearl's [2000:248] early graphical IV definition rules out conditioning on descendants of treatment outright.
 
Since conditioning on a variable can result from many different procedures during data collection or data analysis, selection bias in IV analysis can result from many different procedures as well. 
Analysts should be aware, however, that different ways of conditioning on a variable may induce quantitatively different selection biases.  
In this paper, we contrast selection bias resulting from two empirically common conditioning procedures: sample truncation and covariate adjustment.
 
\emph{Truncation} occurs when observations are preferentially excluded from the sample (\citealt{Bareinboim2014}), e.g. due to attrition or listwise deletion of missing data. 
Write $ R=1 $ for retained observations, and $ R=0 $ for excluded (truncated) observations. 
Let $ S $ be the (possibly latent) continuous variable that determines truncation.
We distinguish between interval truncation and point truncation. 
\emph{Interval truncation} restricts the sample to observations with a range of values of $ S $, for example, $ R=\mathbf{1}(S\geq s_0) $ or $ R=\mathbf{1}(s_1\geq S\geq s_0) $, where $ \mathbf{1}(\cdot ) $ is the indicator function. A limiting case of interval truncation is \emph{point truncation}, where the sample is restricted to units with a single value of $ S $, $ R=\mathbf{1}(S=s_0) $. 
The truncated IV estimator is given by
\[ \beta_{IV|Tr}=\frac{Cov(Z,Y|R=1)}{Cov(Z,T|R=1)}. \]
With truncation (as opposed to censoring) the analyst does not have access to the truncated observations, cannot estimate the probability of truncation, and hence cannot use inverse-probability weights to correct for truncation (\citealt{CananEtAl2017}; \citealt{GkatzionisBurgess2018}). 
In Figure \ref{fig:IVDAGBaseline}, a truncated sample would involve the empiricist observing $ \{Z,T,Y\} $ only for units with $ R=1 $.

Although selection can also occur due to \emph{covariate adjustment} for S, this procedure has received less attention in the literature on IV selection bias. 
With covariate adjustment the analyst observes $ \{Z, T, S, Y\} $ for all units. 
Adjustment involves first exactly stratifying on $ S $, computing the estimator within each stratum, and then averaging across the marginal distribution of $ S $.
Thus the IV estimator under adjustment on $ S $ is given by
\[ \beta_{IV|Adj}=\int \frac{Cov(Z,Y|S=s)}{Cov(Z,T|S=s)}f_S(s)ds, \]
where $f_S(s)$ is the marginal distribution of $ S $.
In linear models, controlling for a variable as a main effect in OLS or 2SLS amounts to covariate adjustment on the variable (\citealt{AngristPischke2008}).

Next, we analytically characterize selection bias in IV analysis and OLS regression for various data-generating models and provide intuition.

\section{Selection Bias in IV: Quantitative Analysis}\label{sec:quantitative}

This section derives exact analytic expressions for  selection  bias  across  a  range  of  common  data-generating  models. For each model, we contrast the selection bias for  the IV and the OLS estimators, resulting from two different conditioning strategies. First, we present the selection bias resulting from covariate adjustment on $ S $. 
Next, we newly derive the selection bias from interval truncation on $ S $, $ R=\mathbf{1}(S\geq s_0 ) $. 
We assume a probit link between $ S $ and the binary selection indicator, $ R $.\footnote{Numerical simulations in prior work have assumed logit selection (\citealt{CananEtAl2017}; \citealt{HughesEtAl2019};  \citealt{GkatzionisBurgess2018}). Switching to probit selection captures the same intuition, but gains analytic tractability.}
Since IV analysis suffers small-sample bias regardless of selection, we study its large-sample behavior (asymptotic bias). 

\subsection{Selection as a Function of Treatment Alone}
Consider the most basic scenario of IV selection bias in Figure \ref{fig:IVDAGBaseline}. 
As stated above, $ Z $ in this model is a valid instrumental variable for the causal effect of $ T $ on $ Y $, $ \beta $ , if the analysis does not condition on $ S $. 
Conditioning on $ S $, however, invalidates $ Z $ as an instrumental variable, because $ S $ is a descendant of $ T $, and $ T $ is a collider on the path $ Z\rightarrow T\leftarrow U\rightarrow Y $. 
Conditioning on $ S $ opens this path, which induces an association between $ Z $ and $ Y $ via $ U $ and hence violates the exclusion condition.
 
Proposition \ref{prop:IVAdjBias} gives the selection bias in the standard IV estimator when the analysis adjusts for S. 

\begin{proposition}\label{prop:IVAdjBias}
	In a linear and homogeneous model with normal errors represented by Figure \ref{fig:IVDAGBaseline} and covariate adjustment on $ S $, the standard instrumental variables estimator converges in probability to 
	\[ 	
	\beta_{IV|Adj}=\beta -\delta_1 \delta_2  \frac{\gamma^2}{1-\gamma ^2}.
	\]
\end{proposition}

The proof follows from regression algebra and Wright's rule (\citealt{Wright1934}). 
The magnitude of selection bias due to covariate adjustment in the IV estimator depends on two components. 
First, selection bias increases with the strength of unobserved confounding between $ T $ and $ Y $ via $ U $, $ \delta_1 \delta_2 $ (which corresponds to the path $ Z\rightarrow T\leftarrow U\rightarrow Y $ that is opened by conditioning on $ S $, less the first stage $ Z\rightarrow T $). 
Second, selection bias increases with the effect of the treatment $ T $ on the selection variable, $ S $ , $ \gamma $ . 
When $ \gamma =0 $, $ S $ contains no information about the collider $ T $, conditioning on $ S $ does not open the path $ Z\rightarrow T\leftarrow U\rightarrow Y $, and selection bias is zero. 
By contrast, as $ |\gamma |\rightarrow 1 $, the magnitude of the bias increases without bound because adjusting for $ S $ increasingly amounts to adjusting for the collider $ T $ itself, while at the same time reducing the first stage. 
(If the analysis directly adjusted for $ T $, then the first stage would go to zero and the IV estimator would not be defined.)

Proposition \ref{prop:IVTruncBias} derives the IV selection bias due to interval truncation on S.
\begin{proposition}\label{prop:IVTruncBias}
	In a linear and homogeneous model with normal errors represented by Figure \ref{fig:IVDAGBaseline} and truncation on $ S $, $ R=\mathbf{1}(S\geq s_0) $, the standard instrumental variables estimator converges in probability to 
	\[ 
	\beta_{IV|Tr}=\beta -\delta_1 \delta_2  \frac{\psi \gamma ^2}{1-\psi \gamma^2}, \quad \text{where }  \psi =\frac{\phi (s_0)}{1-\Phi(s_0)} \left( \frac{\phi (s_0 )}{1-\Phi(s_0 )} - s_0 \right),
	\]
	and $ \phi (\cdot) $ and $ \Phi(\cdot ) $ are the standard normal pdf and cdf, respectively. 
\end{proposition}

Proposition \ref{prop:IVTruncBias} (proved in Appendix \ref{sec:TruncBiasProof}) illustrates that IV selection bias due to truncation (Proposition \ref{prop:IVTruncBias}) differs from IV selection bias due to adjustment (Proposition \ref{prop:IVAdjBias}) only in that truncation deflates the contribution of the effect of $ T $ on $ S $, $ \gamma $, by the factor $ \psi \in (0,1) $.
Since $ \psi\ $  is the derivative of the standard normal hazard function, it monotonically increases with the \emph{severity of truncation}, $Pr(R=0)=\Phi(s_0)$, as shown in Figure \ref{fig:PsiVsSelection}. Hence, interval truncation leads to less IV selection bias than covariate adjustment in Figure \ref{fig:IVDAGBaseline},

\begin{figure}[t!]
	\centering
	\begin{subfigure}{0.48\linewidth}
	    \centering
	    \textbf{\qquad Truncation Severity versus $ \psi $}
		\includegraphics[width=\linewidth]{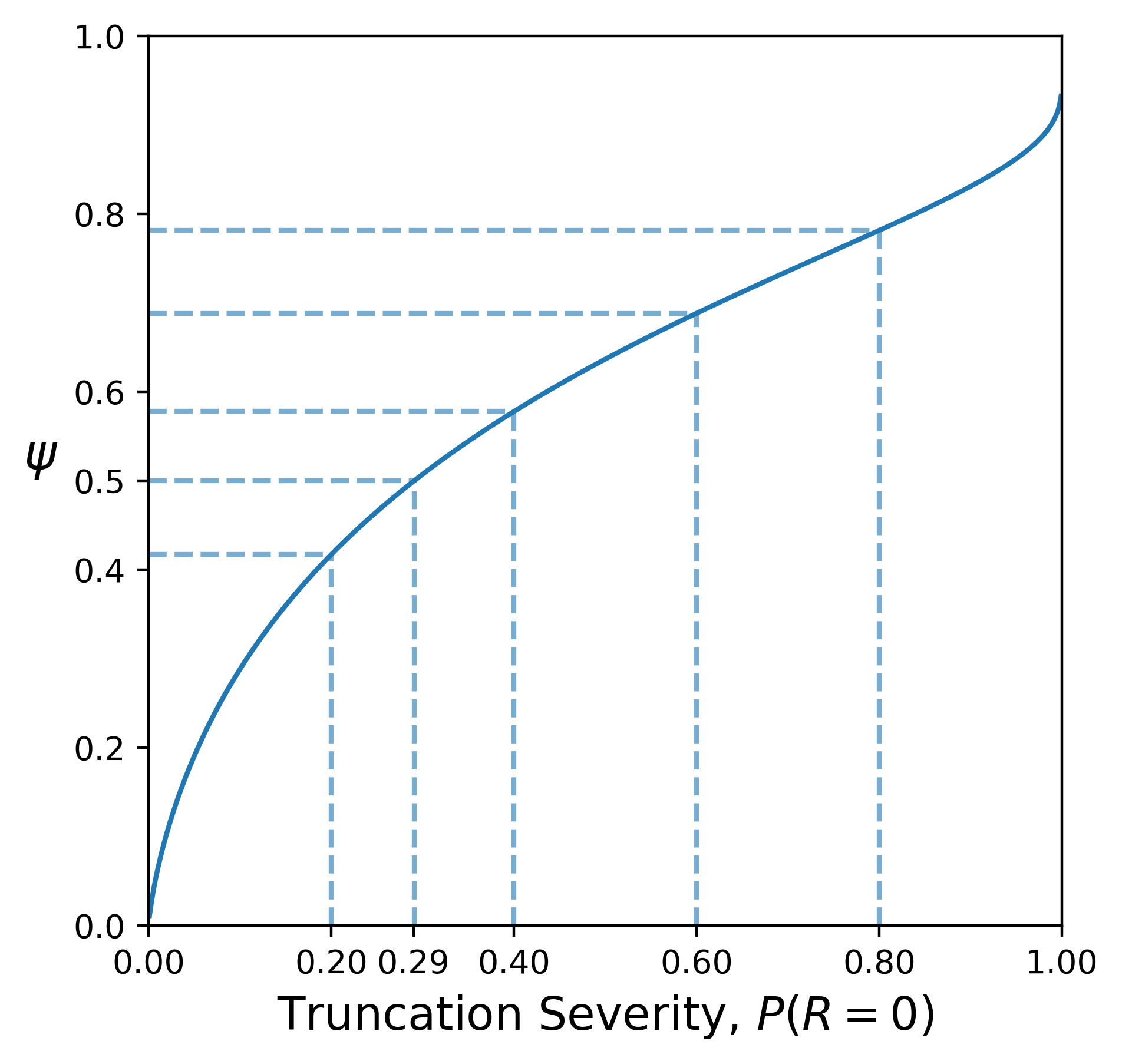}
		\caption{}
		\label{fig:PsiVsSelection}
	\end{subfigure}
	\begin{subfigure}{0.50\linewidth}
	    \centering
	    \textbf{\qquad Least Biased Estimator}\\
		\includegraphics[width=\linewidth]{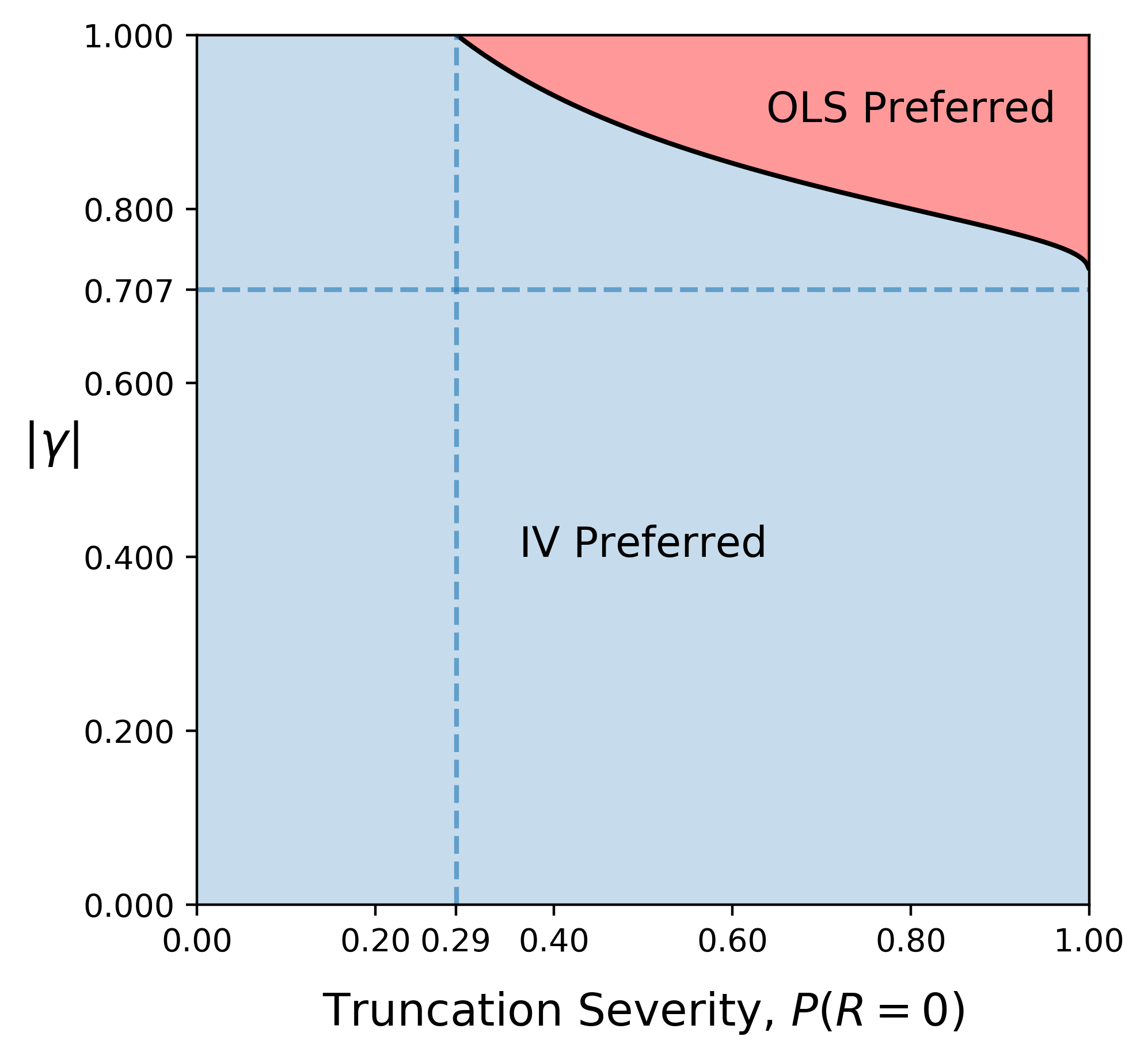}
		\caption{}
		\label{fig:LeastBiasEstPreference}
	\end{subfigure}
	\caption{(\subref{fig:PsiVsSelection}) $\psi$ monotonically increases with truncation severity. (\subref{fig:LeastBiasEstPreference}) Whether OLS or IV is less biased under selection depends on truncation severity and the effect of $ T $ on $ S $, $ |\gamma| $.}
	\label{fig:PsiPlots}
\end{figure}


\begin{corollary}\label{cor:IVTruncIVAdj}
	In a linear and homogeneous model with normal errors represented by Figure \ref{fig:IVDAGBaseline}, the magnitude of IV-adjustment bias is weakly larger than that of IV-truncation bias: 
	$
	\left| \beta_{IV|Adj} - \beta \right| \geq \left| \beta_{IV|Tr} - \beta \right|.
	$
\end{corollary}
Corollary \ref{cor:IVTruncIVAdj} makes intuitive sense. 
Adjustment involves first exactly stratifying and then averaging across strata defined by $ S=s $. 
Exact stratification on $ S $ uses all information about $ T $ that is contained in $ S $, hence opening the biasing path as much as conditioning on $ S $ possibly can.  
By contrast, interval truncation amounts to imprecise stratification on $ S $ (retaining observations across a range of values on $ S $, but not exactly stratifying on any particular value), hence ``less opening" the biasing path.

Of some methodological interest, we further note, in Figure \ref{fig:IVDAGBaseline}, that IV selection bias by truncation converges on IV selection bias by covariate adjustment as the severity of truncation increases to shrink the remaining sample to a single point. 
Proposition \ref{prop:AdjIsPtTrunc} states that this observation is true for all models, not only for Figure \ref{fig:IVDAGBaseline}.
\begin{proposition}\label{prop:AdjIsPtTrunc}
	In a linear and homogeneous model with normal errors, selection bias in the standard instrumental variables estimator due to covariate adjustment is the limiting case of selection bias due to point truncation, 
	\[ 
	\lim_{s_0\rightarrow \infty}\beta_{IV|Tr} = \beta_{IV|Adj}.   
	\]
\end{proposition}

This proposition makes intuitive sense. 
Covariate adjustment involves exact stratification on $ S=s $, which defines point truncation. 
Since the probability limits of all $s$-stratum specific estimators are identical in linear Gaussian models, selection bias by adjustment equals selection bias by point truncation.  
The proof in Appendix \ref{sec:TruncAdjProof} formalizes this intuition. 

Proposition \ref{prop:IVTruncBias} helps inform empirical choices in practice. When selection is unavoidable (e.g. because the data were truncated during data collection), should analysts choose IV or OLS? 
Figure \ref{fig:LeastBiasEstPreference} shows that the IV estimator is preferred to OLS, with respect to bias, for most combinations of $ \gamma  $ and truncation severity.
Since OLS bias (with or without truncation) only depends on unobserved confounding, i.e. $ \beta _{OLS|Tr} - \beta=\delta_1 \delta_2 $, the difference in magnitude between the OLS and IV biases with truncation is given by 
\[ 
\left|\beta_{OLS|Tr}-\beta\right| - \left|\beta_{IV|Tr} - \beta \right|=\left|\delta_1 \delta_2 \right|  \frac{1-2\psi \gamma^2}{1-\psi \gamma^2}.
\]
Hence, the IV estimator is preferred when $ \psi \gamma^2 \leq \frac{1}{2} $.
Specifically, when fewer than 29.1\% of observations are truncated (corresponding to $ \psi \leq 0.5 $), IV is preferred regardless of the effect of $ T $ on $ S $, $ \gamma $. 
Conversely, when $ |\gamma |<\sqrt{0.5}\approx 0.707 $, no amount of truncation makes OLS preferable over IV. 
Recalling that $ \gamma $ cannot exceed 1 in magnitude, the selection variable $ S $ would have to be an extraordinarily strong proxy for $ T $ to make IV more biased than OLS at any level of truncation.

Perhaps most useful for practice, we note that selection bias (by truncation or adjustment) in Figure \ref{fig:IVDAGBaseline} is proportional to the negative of OLS confounding bias. 
Therefore, the OLS and IV estimators under selection bound the true causal effect.

\begin{corollary}
	In a linear and homogeneous model with normal errors represented by Figure \ref{fig:IVDAGBaseline}, the OLS estimator and the instrumental variables estimator with selection bound the causal effect of $ T $ on $ Y $, $ \beta $ ,
	\begin{align*}
        \beta_{IV|Tr}\leq & \beta \leq \beta_{OLS}, \quad \text{when } \quad \delta_1 \delta_2>0, \\
	    \beta_{IV|Tr}\geq & \beta \geq \beta_{OLS}, \quad \text{when } \quad \delta_1 \delta_2<0. 	    
	\end{align*} 
\end{corollary}

The fact that the IV selection bias has the opposite sign of the OLS selection bias in Figure \ref{fig:IVDAGBaseline} is owed to linearity and homogeneity: in linear and homogeneous models, conditioning on a collider or its descendant reverses the sign of the product of the path parameters for the associated path. 
For example if all path parameters along the biasing path $ Z\rightarrow T\leftarrow U\rightarrow Y $ are positive, then conditioning on $ S\in desc(T) $ will induce a negative association along this path. 
Since the IV bias hinges on conditioning on $ S $, the selection bias would be negative. 
By contrast, OLS bias in Figure \ref{fig:IVDAGBaseline} does not hinge on conditioning on $ S $ and instead results from confounding along $ T\leftarrow U\rightarrow Y $. 
Therefore, OLS bias would be positive. 

\subsection{Selection as a Function of a Mediator}\label{sec:medselection}

Next, consider models in which the selection variable, $ S $, is a mediator of the effect of treatment on the outcome, as in the causal graphs in Figures \ref{fig:DAGTandY} and \ref{fig:DAGTandYandV}. 
These situations are worth investigating for two reasons: first, often empiricists are interested in the direct causal effect of $ T $ on $ Y $, which necessitates conditioning on $ S $; 
second, they result in qualitatively different bias representations.

\begin{figure}[b!]
	\centering
	\begin{subfigure}{0.48\linewidth}
		\usetikzlibrary{arrows}
\def\xscale{1.2}
\def\yscale{1.5}
\tikzstyle{DAGarrow} = [-latex]
\begin{tikzpicture} 
	\node (z) at (-2*\xscale,0) {$Z$};
	\node (t) at (0,0) {$T$};
	\node (y) at (2*\xscale,0) {$Y$};
	\node (u) at (1*\xscale,1*\yscale) {$U$};
	\node (s) at (0,-1*\yscale) {$S$};
	
	\draw [DAGarrow] (z) -- node [above] {$\pi$} (t);
	\draw [DAGarrow] (s) -- node [below right] {$\tau$} (y);
	\draw [DAGarrow] (t) -- node [above] {$\beta$} (y);
	\draw [DAGarrow] (u) -- node [above left] {$\delta_1$} (t);
	\draw [DAGarrow] (u) -- node [above right] {$\delta_2$} (y);
	\draw [DAGarrow] (t) -- node [left] {$\gamma$} (s);
\end{tikzpicture}
		\caption{}
		\label{fig:DAGTandY}
	\end{subfigure}
	\begin{subfigure}{0.48\linewidth}
\def\xscale{1.2}
\def\yscale{1.5}
\tikzstyle{DAGarrow} = [-latex]
\begin{tikzpicture} 
	\node (z) at (-2*\xscale,0) {$Z$};
	\node (t) at (0,0) {$T$};
	\node (y) at (2*\xscale,0) {$Y$};
	\node (u) at (1*\xscale,1*\yscale) {$U$};
	\node (s) at (0,-1*\yscale) {$S$};
	\node (w) at (2*\xscale,-1*\yscale) {$W$};
	
	\draw [DAGarrow] (z) -- node [above] {$\pi$} (t);
	\draw [DAGarrow] (s) -- node [below right] {$\tau$} (y);
	\draw [DAGarrow] (t) -- node [above] {$\beta$} (y);
	\draw [DAGarrow] (u) -- node [above left] {$\delta_1$} (t);
	\draw [DAGarrow] (u) -- node [above right] {$\delta_2$} (y);
	\draw [DAGarrow] (t) -- node [left] {$\gamma$} (s);
	\draw [DAGarrow] (w) -- node [below] {$\delta_3$} (s);
	\draw [DAGarrow] (w) -- node [right] {$\delta_4$} (y);
\end{tikzpicture}
		\caption{}
		\label{fig:DAGTandYandV}
	\end{subfigure}
	\caption{IV scenarios where the selection variable is both a descendant of treatment and a mediator.}
	\label{fig:IVDAGMediated}
\end{figure}
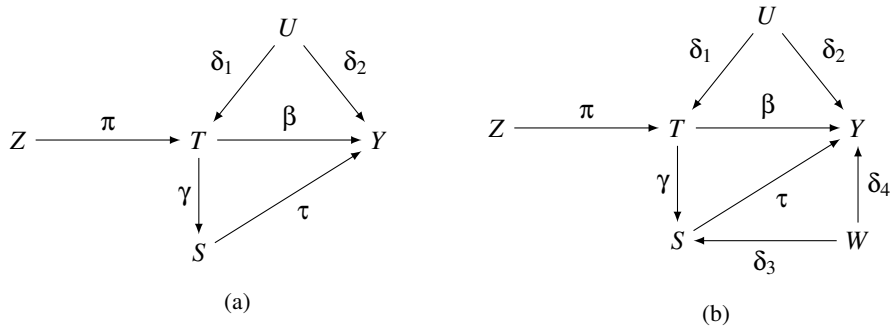

Suppose that the analyst is interested in the direct causal effect of $ T $ on $ Y $, $ \beta $, in the model of Figure \ref{fig:DAGTandY}. 
The bias in the IV and OLS estimators under interval truncation and adjustment for $ S $ is given in Proposition \ref{prop:IVMedBias}. 

\begin{proposition}\label{prop:IVMedBias}
	In a linear and homogeneous model with normal errors represented by Figure \ref{fig:DAGTandY}. The standard instrumental variables estimator with selection on $ S $, converges in probability to 
	\[ 
	\beta_{IV|S}= \beta -\delta_1 \delta_2  \frac{\psi \gamma ^2}{1-\psi \gamma ^2} + \gamma\tau \frac{1-\psi}{1-\psi \gamma ^2},
	\]
	and the OLS estimator with selection on $ S $ converges in probability to 
	\[ 
	\beta_{OLS|S}= \beta +\delta_1 \delta_2+\gamma \tau \frac{1-\psi}{1-\psi \gamma ^2},
	\]
	where 
	\[ 
	\psi =\begin{cases}
				\frac{\phi (s_0)}{1-\Phi(s_0)} \left( \frac{\phi (s_0 )}{1-\Phi(s_0 )} - s_0 \right) & \text{with truncation on }S, R=\mathbf{1}(S\geq s_0) \\
				1 & \text{with adjustment on }S 
		\end{cases}.
	\]
\end{proposition}
All bias expressions in Proposition \ref{prop:IVMedBias} have a straightforward graphical interpretation. 
With \emph{adjustment} on $ S $, the indirect causal path $ T\rightarrow S\rightarrow Y $ is completely blocked, because $ S $ is a non-collider on this path. 
Hence, the bias in the IV and OLS estimators with adjustment on $ S $ equals the IV and OLS adjustment biases in Figure \ref{fig:IVDAGBaseline}, where $ S $ was not a mediator.
With adjustment on $ S $, IV is biased by selection, whereas OLS is biased by confounding; 
IV selection bias will generally be smaller in magnitude than OLS confounding bias (unless the effect of $ T $ on $ S $ is very large); and IV and OLS with adjustment bound the true direct causal effect. 

With \emph{truncation} on $ S $, however, the indirect path $ T\rightarrow S\rightarrow Y $ is not completely blocked and hence contributes a new term to both IV and OLS bias. For both IV and OLS, this term equals the strength of the partially blocked indirect path, $ \gamma \tau $ , deflated by the multiplier $ 0\leq(1-\psi)/(1-\psi \gamma^2 )\leq 1 $. 
The size of the multiplier depends both on the truncation severity, $ \psi $ , and on the effect of $ T $ on $ S $, $ \gamma $ , but in opposite directions.
As $ \gamma $ is fixed and truncation increases, $ \psi \rightarrow 1 $, the analysis conditions ever more precisely on an ever smaller range of values of $ S $; hence the indirect path is increasingly blocked, and both the multiplier and the bias term tend to 0. 
By contrast, when $ \psi $ is fixed and the effect of $ T $ on $ S $ increases, $ |\gamma| \rightarrow 1 $, the information about $ T $ contained in $ S $ increases, the multiplier tends to 1, and the path is increasingly opened.
 
By Proposition \ref{prop:AdjIsPtTrunc}, it remains true in Figure \ref{fig:DAGTandY} that IV selection bias due to adjustment is the limiting case of IV selection bias due to point truncation. 
However, it is no longer necessarily true that IV with adjustment is more biased than IV with truncation. 
The bias ordering now depends on the signs and relative sizes of the two additive bias term (representing the biasing paths $ T\leftarrow U\rightarrow Y $ and $ T\rightarrow S\rightarrow Y $), and on how well the indirect path $ T\rightarrow S\rightarrow Y $ is closed by truncation. 
Hence, when selection is made on a mediator of the treatment effect, selection bias by adjustment could be larger or smaller in magnitude than selection bias by truncation.
 Bounding the true causal effect also becomes more difficult. 
With truncation on $ S $, IV and OLS with selection do not necessarily bound the true direct causal effect. 

The analysis is further complicated when the effect of $ S $ on $ Y $ is confounded by some unobserved variable, $ W $, as in Figure \ref{fig:DAGTandYandV}. 
This situation is arguably more realistic than the model in Figure \ref{fig:DAGTandY}, because mediators in observational studies are expected to be confounded. 
Here, conditioning on $ S $ (by adjustment or truncation) in IV analysis opens a new path, $ Z\rightarrow T\rightarrow S\leftarrow W\rightarrow Y $, which violates the exclusion assumption; and in OLS it opens $ T\rightarrow S\leftarrow W\rightarrow Y $, which biases OLS regression. 
The resulting bias expressions are the same as those in Proposition \ref{prop:IVMedBias} with an additional bias term, $ -\gamma \delta_3 \delta_4  \frac{\psi}{1-\psi \gamma^2} $. 
Once more, IV selection bias due to adjustment is the limiting case of IV selection bias due to point truncation. 
However, no pair of estimators (among $ \beta_{IV|Tr}, \beta_{IV|Adj},\beta_{OLS|Tr},\beta_{OLS|Adj} $) can be relied on to bound the true direct causal effect in the model of Figure \ref{fig:DAGTandYandV}.

\subsection{Selection on Treatment and the Unobserved Confounder}\label{sec:medandconfselection}

Finally, we consider situations where the selection variable, $ S $, is also a descendant of the unobserved $ U $ that confounds the effect of treatment on the outcome (Figure \ref{fig:DAGTandU}).

\begin{figure}[t!]
	\centering
\def\xscale{1.2}
\def\yscale{1.5}
\tikzstyle{DAGarrow} = [-latex]
\begin{tikzpicture} 
	\node (z) at (-2*\xscale,0) {$Z$};
	\node (t) at (0,0) {$T$};
	\node (y) at (2*\xscale,0) {$Y$};
	\node (u) at (1*\xscale,1*\yscale) {$U$};
	\node (s) at (0,-1*\yscale) {$S$};
	\node (j) at (2.4*\xscale,0) {$ $};

	\draw [DAGarrow] (z) -- node [above] {$\pi$} (t);
	\draw [DAGarrow] (u.east) to [out=0, in=90] node [above] {$\delta_3$} (j) to [out=-90, in=0] (s.east);
	\draw [DAGarrow] (t) -- node [above] {$\beta$} (y);
	\draw [DAGarrow] (u) -- node [above left] {$\delta_1$} (t);
	\draw [DAGarrow] (u) -- node [above right] {$\delta_2$} (y);
	\draw [DAGarrow] (t) -- node [left] {$\gamma$} (s);
\end{tikzpicture}
	\caption{IV scenario where the selection variable is both a descendant of the treatment and the unobserved confounder.}
	\label{fig:DAGTandU}
\end{figure}
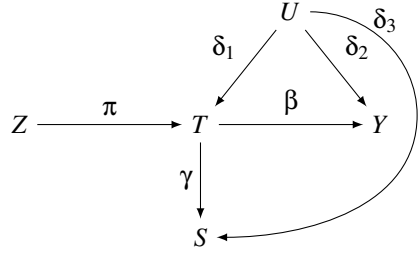


\begin{proposition}\label{prop:IVUnobsBias}
	In a linear and homogeneous model with normal errors represented by Figure \ref{fig:DAGTandU}. The standard instrumental variables estimator with selection on $ S $ converges in probability to
	\[ 	
	\beta_{IV|S} = \beta -\delta_1 \delta_2  \frac{\psi \gamma ^2}{1-\psi \gamma (\gamma +\delta_1 \delta_3 )} - \gamma\delta_3\delta_2\frac{\psi}{1-\psi \gamma (\gamma +\delta_1 \delta_3 )}, 
	\]
	and the OLS estimator with selection on S converges in probability to 
	\[ 
	\beta_{OLS|S} = \beta +\delta_1\delta_2  \frac{1-\psi  (\gamma^2+\gamma \delta_1 \delta_3+\delta_3^2)}{1-\psi\gamma(\gamma+\delta_1 \delta_3 )^2} - \gamma \delta_3 \delta_2\frac{\psi}{1-\psi \gamma (\gamma +\delta_1 \delta_3 )^2},
	\]
	where 
	\[ 
	\psi =\begin{cases}
	\frac{\phi (s_0)}{1-\Phi(s_0)} \left( \frac{\phi (s_0 )}{1-\Phi(s_0 )} - s_0 \right) & \text{with truncation on }S,R=\mathbf{1}(S\geq s_0) \\
	1 & \text{with adjustment on }S 
	\end{cases}.
	\]
\end{proposition}

Three points stand out about selection bias in Figure \ref{fig:DAGTandU}. 
First, when $ S $ is a descendant of both $ T $ and $ U $, conditioning on $ S $ opens a new path, $ T\rightarrow S\leftarrow U\rightarrow Y $, which biases IV and OLS with adjustment or truncation on $ S $.
 
 Second, in contrast to models considered previously, the bias term associated with each biasing path ($ T\leftarrow U\rightarrow Y $ and $ T\rightarrow S\leftarrow U\rightarrow Y $) is now a function of the path parameters of both paths.  In other words, the path-specific biases interact. Pearl's graphical causal models provide  intuition for this interaction.
Consider, for example, the second bias term. First, conditioning on $ S $ opens the path $ T\rightarrow S\leftarrow U\rightarrow Y $. Hence, the bias term depends on $ \gamma \delta_3 \delta_2 $. Second, conditioning on $ S $ also absorbs variance from $ U $ (a non-collider on $ T\rightarrow S\leftarrow U\rightarrow Y $), because $S$ is a descendant of $U$ along the path $U\rightarrow T\rightarrow S$.  Hence, the bias term also depends on $ \delta_1 $. 

Third, the direction of the interaction, and hence the overall bias, depends on the specific parameter values. This makes the bias order of these estimators fairly unpredictable and prevents generic recommendations for or against any one estimator. This ambiguity provides additional motivation for using exact bias formulas for sensitivity analysis.

\section{Conclusion}\label{sec:conclusion}
Conditioning on the wrong variable can induce selection bias in IV analysis. 
When consistent estimators are not available, analysts should gauge the bias in their estimators by principled speculation or formal sensitivity analysis. 
To enable this work, we have derived analytic expressions for IV selection biases that have previously been characterized only by simulation. 

Our analysis specifically focused on scenarios where selection is a function of a confounded treatment. Judea Pearl's [2000] graphical IV criterion specifically prohibited conditioning on a descendant of treatment. But the practice appears to remain common, calling for formal analysis. 
Our analytic expressions present asymptotic IV selection bias in terms of substantively interpretable standardized path parameters for Gaussian models. 
Empowered by Pearl's graphical causal models, we further provided intuition by decomposing the bias into terms that map onto the paths in the data-generating model that are opened (or closed) by selection. 
Leveraging prior knowledge or principled theory, analysts may use our bias expressions to conduct formal sensitivity analyses by populating the free parameters to derive the size of the bias. Even with partial information our expressions may provide informative bounds on the bias.

We present three broad conclusions. 
First, in the models we investigated, IV selection bias depends on three ingredients: (i) the strength of each biasing path in the model, (ii) the effect of treatment on the selection variable, $ |\gamma| $; and (iii) truncation severity, $\psi$, i.e. the share of the full sample excluded from the analysis by truncation.  
The magnitude of the bias term associated with each biasing path increases with the strength of the path, with $ |\gamma| $, and with truncation severity, $\psi$, if selection is made on a collider or descendant of a collider on the path; and the magnitude of the bias term increases with the strength of the path and with $ |\gamma| $, but decreases with $ \psi $, if selection is made on a non-collider on the biasing path. 

Second, sign and magnitude of IV selection bias depend on the selection procedure: in all linear Gaussian IV models, the bias induced by covariate adjustment is the limiting case of bias induced by point truncation. 
This does not mean that adjustment bias is always larger than truncation bias, only that adjustment bias equals truncation bias if truncation had reduced the sample to a single point. 

Third, rather usefully, in some models (where selection is only a function of treatment and the selection variable is not a mediator), IV and OLS suffer selection biases of opposite signs, such that these estimators bound the true causal effect. 
In the same models, unless the effect of treatment on selection is very large, IV with selection suffers less bias than OLS with or without selection.

\bibliographystyle{mcp-acm}
\bibliography{bibliography}  

\section{Appendix}\label{sec:Appendix}
\subsection{Proof of Truncation Bias Expressions}\label{sec:TruncBiasProof}
We derive the bias under truncation by leveraging a result from  \citet{Tallis1965}.
\begin{lemma} \label{lem:Tallis}
	Let $V\in \mathbb{R}^k$ follow a multivariate normal distribution, $ V \sim N \left( 0, \Sigma \right) $, and define the truncated random vector $ \widetilde{V} = \{ v\in V: c'v \geq p \} $ with $ p\in \mathbb{R} $,$ c\in\mathbb{R}^k $, and $ |c|=1 $.
	Then the expectation and variance of the truncated random vector are given by
	\begin{align*}
	E\left[ \widetilde{V} \right] & = \Sigma c \kappa^{-1} \lambda \left( \frac{p}{\kappa}\right)\\
	Var\left( \widetilde{V} \right) & = \Sigma - \Sigma cc' \Sigma \kappa^{-2}\psi 
	\end{align*}
	where $ \kappa = \left(c'\Sigma c \right)^{-1/2} $, $ \lambda(x) = \frac{\phi(x)}{1-\Phi(x)}$ is the hazard function of the standard normal distribution, and 
	\[
	\psi = \lambda \left( \frac{p}{\kappa}\right)\left( \lambda \left( \frac{p}{\kappa}\right) - \frac{p}{\kappa} \right).
	\]
\end{lemma}
Using properties of the standard normal hazard function it can be shown that $\psi$ is in fact the derivative of the hazard function.

\begin{proof}[Proof of Proposition \ref{prop:IVTruncBias}]
	Consider the model described by Figure \ref{fig:IVDAGBaseline}. 
	Since the idiosyncratic shocks are all normally distributed, all variables in the model are normally distributed.
	Specifically for vectors $ V = \begin{bmatrix} Z & U & T & S & Y \end{bmatrix}' $ and $ \varepsilon = \begin{bmatrix} \varepsilon_Z & \varepsilon_U & \varepsilon_T & \varepsilon_S & \varepsilon_Y \end{bmatrix}' $,
	 the standardized\footnote{Standardization implies non-unit variance for some of the shocks. For example when $ Var(T)=1 $, then $ \varepsilon_{T} $ is $ Var(\varepsilon_T) = 1-\pi^2-\delta_1^2 $.} model has the reduced form $ V = \Gamma \varepsilon $, where $ \varepsilon \sim N(0,\Sigma_{\varepsilon}) $ and
	\[ 
	\Gamma = \begin{bmatrix}
	1 & 0 & 0 & 0 & 0 \\
	0 & 1 & 0 & 0 & 0 \\
	\pi & \delta_{1} & 1 & 0 & 0 \\
	\gamma \pi & \gamma \delta_{1} & \gamma & 1 & 0 \\
	\beta \pi & \beta \delta_{1} + \delta_{2}& \beta & 0 & 1 \\
	\end{bmatrix} \qquad 
	\Sigma_{\varepsilon}= \begin{bmatrix}
	1 & 0 & 0 & 0 & 0 \\
	0 & 1 & 0 & 0 & 0 \\
	0 & 0 & 1-\pi^2-\delta_1^2 &  0 & 0 \\
	0 & 0 & 0 & 1-\gamma^2 &  0 \\
	0 & 0 & 0 & 0 & 1-\beta^2-\delta_2^2-2\beta\delta_1\delta_2 \\
	\end{bmatrix}.
	\]
	
	Since this implies that $ V \sim N\left(0, \Gamma \Sigma_{\varepsilon}\Gamma' \right) $, our truncation scenario, $ R = \mathbf{1}(S\geq s_0) $, allows for direct application of Lemma \ref{lem:Tallis} to derive the covariance matrix of the truncated distribution, $ \widetilde{V} = V|S\geq s_0 $.
	For Lemma \ref{lem:Tallis}, $ c= \begin{bmatrix} 0 & 0 & 0 & 1 & 0 \end{bmatrix}' $, $ p=s_0 $, and $ \Sigma = \Gamma \Sigma_{\varepsilon}\Gamma' $.
	This implies $ \kappa = 1 $ and thus
	\[ 
	Var\left( \widetilde{V} \right) = \Gamma \Sigma_{\varepsilon}\Gamma' - \Gamma \Sigma_{\varepsilon}\Gamma' cc' \Gamma \Sigma_{\varepsilon}\Gamma'\psi 
	\quad \text{where} \quad \psi = \lambda( s_0 )\left( \lambda( s_0 ) - s_0 \right).
	\]
	
	Finally the IV estimand with truncation is given by the ratio of the truncated covariance between instrument and outcome and the truncated covariance between instrument and treatment. After some enjoyable algebra, we evaluate $ Var(\widetilde{V}) $, extract the relevant covariances, and obtain
	\begin{align*}
	\beta_{IV|Tr} & = \frac{Cov(Z,Y|S\geq s_0)}{Cov(Z,T|S\geq s_0)} = \frac{\beta\pi - \psi \gamma\pi\left( \beta\gamma + \gamma\delta_1\delta_2 \right)}{\pi - \psi\gamma^2\pi} = \beta - \delta_1\delta_2\frac{\psi\gamma^2}{1-\psi\gamma^2}.
	\end{align*}
\end{proof}

The proofs of Propositions \ref{prop:IVMedBias} and \ref{prop:IVUnobsBias} proceed analogously, using the appropriate reduced form matrix, $ \Gamma $, for each scenario.

\subsection{Proof of Adjustment as Point Truncation (Proposition \ref{prop:AdjIsPtTrunc})}\label{sec:TruncAdjProof}
\begin{proof}
    Define the stratum specific IV estimator when $S=s$ as 
    \[
    \beta_{IV|S}\left(s\right)=\frac{Cov\left(Z,Y|S=s\right)}{Cov\left(Z,T|S=s\right)}
    \]
    Notice that $\beta_{IV|S}\left(s\right)$ is the IV estimator under point truncation (i.e. the limit of the interval truncated estimator as the interval collapses to a point).
    
    In a homogeneous linear model with normal errors, $ V = \begin{bmatrix} Z & U & T & S & Y \end{bmatrix}' $ will follow a multivariate normal distribution.
    Multivariate normal distributions have the useful property that their conditional distributions have constant covariances across the conditioning level.
    Hence, for all $ V_1, V_2, V_3 \in \{Z, U, T, S, Y\} $ and $ v_0, v_1 \in \mathbb{R} $, we have that
    \[ 
    Cov(V_1, V_2| V_3 = v_0)  = Cov(V_1, V_2| V_3 = v_1). 
    \] 
    It follows that $\beta_{IV|S}\left(s_{0}\right)=\beta_{IV|S}\left(s_{1}\right)$
    for any $s_{0},s_{1}\in\mathbb{R}$.
    Since the stratum specific IV estimator is constant across strata of $ S $, this implies that the IV estimator under adjustment on $ S $ is the same as any stratum specific IV estimator. 
\end{proof}

\end{document}